\documentclass[11pt]{article}
\usepackage{times}
\usepackage{amsmath}
\usepackage{amssymb}
\usepackage{amsthm}
\usepackage{natbib}
\usepackage{nicefrac}
\usepackage{fullpage}

\usepackage{tikz}
\usetikzlibrary{arrows,positioning,calc}

\usepackage{xcolor}
\usepackage{soul}
\definecolor{myRed}{rgb}{0.5,0.2,0.2}

\makeatletter
\renewcommand\bibsection%
{
  \noindent{\large\textbf{\refname
    \@mkboth{\MakeUppercase{\refname}}{\MakeUppercase{\refname}}}}
}
\makeatother

\newtheorem{theorem}{Theorem}[section]
\newtheorem{lemma}[theorem]{Lemma}

\newtheorem{definition}[theorem]{Definition}

\newtheorem{remark}{Remark}[section]

\title{Nonuniform Graph Partitioning \\with Unrelated Weights}
\author{Konstantin Makarychev\\Microsoft Research\and Yury Makarychev\thanks{Supported by NSF CAREER award CCF-1150062 and NSF grant IIS-1302662.}\\Toyota Technological Institute at Chicago}
\date{}

\newcommand{\OpenFrame}{\rule{0pt}{12pt} \hrule height 0.8pt \rule{0pt}{1pt} \hrule height 0.4pt \rule{0pt}{6pt}}
\newcommand{\CloseFrame}{\rule{0pt}{1pt}\hrule height 0.4pt \rule{0pt}{1pt} \hrule height 0.8pt \rule{0pt}{12pt}}

\newcommand {\subsec}[1]{\par{\textbf{#1. }}}

\DeclareMathOperator {\depth}  {depth}

\newcommand {\brc}   [1] {\left(#1\right)}

\newcommand {\Exp}       {\mathbb{E}}
\newcommand {\Prob}  [1] {\Pr \brc{#1 }}

\newcommand{\given}{\mid}
\newcommand{\Given}{\;\mid\;}

\newcommand {\bbR}    {\mathbb{R}}

\newcommand {\calC}   {{\cal{C}}}
\newcommand {\calI}   {{\cal{I}}}

\newcommand {\calP}   {{\cal{P}}}

\newcommand {\calA}   {{\cal{A}}}

\begin{document}
\maketitle
\begin{abstract}
We give a bi-criteria approximation algorithm for the Minimum Nonuniform Partitioning problem, recently introduced by \citet*{KNST}. In this problem,
we are given a graph $G=(V,E)$ on $n$ vertices and $k$ numbers $\rho_1,\dots, \rho_k$. 
The goal is to partition the graph into $k$ disjoint sets $P_1,\dots, P_k$ satisfying 
$|P_i|\leq \rho_i n$ so as to minimize the number of edges cut by the partition.
Our algorithm has an approximation ratio of $O(\sqrt{\log n \log k})$ for general graphs,
and an $O(1)$ approximation for graphs with excluded minors. This is an improvement upon the $O(\log n)$ 
algorithm of~\citet*{KNST}. Our approximation ratio matches
the best known ratio for the Minimum (Uniform) $k$-Partitioning problem. 

We extend our results to the case of ``unrelated weights'' and to the case of ``unrelated $d$-dimensional weights''. 
In the former case, different vertices may have different weights
and the weight of a vertex may depend on the set $P_i$ the vertex is assigned to.
In the latter case, each vertex $u$ has a $d$-dimensional weight 
$r(u,i) = (r_1(u,i), \dots, r_d(u,i))$ if $u$ is assigned to $P_i$. Each set $P_i$ has a $d$-dimensional capacity $c(i) = (c_1(i),\dots, c_d(i))$. The goal is to find a partition such that $\sum_{u\in {P_i}} r(u,i) \leq c(i)$ coordinate-wise.

\end{abstract}

\section{Introduction}
\hyphenation{Krauth-gamer}
We study the Minimum Nonuniform Partitioning problem, which was recently proposed by 
Krauthgamer, Naor, Schwartz and Talwar (2014). 
We are given a graph $G=(V,E)$, parameter $k$ and $k$ numbers (capacities) $\rho_1,\dots, \rho_k$. Our goal
is to partition the graph $G$ into $k$ pieces (bins) $P_1, \dots, P_k$ satisfying capacity constraints $|P_i|\leq \rho_i n$
so as to minimize the number of cut edges.
The problem is a generalization of the Minimum $k$-Partitioning problem studied by
\citet*{KNS}, in which 
all bins have equal capacity $\rho_i = 1/k$.

The problem has many applications (see Krauthgamer et al. \citeyear{KNST}). 
Consider an example in 
cloud computing: Imagine that we need to distribute $n$ computational 
tasks -- vertices of the graph -- among $k$ machines, each with capacity $\rho_i n$. 
Different tasks communicate with each other. 
The amount of communication 
between tasks $u$ and $v$
equals the weight of the edges between the corresponding vertices $u$ and $v$. 
Our goal is to distribute tasks among $k$ machines subject to capacity constraints 
so as to minimize the total amount of communication between machines.\footnote{
In this example, we need to solve
a variant of the problem with edge weights.}

The problem is quite challenging. Krauthgamer et al. (2014) note that many existing techniques do not work for this problem. 
Particularly, it is not clear how to solve this problem on tree graphs\footnote{Our algorithm gives a constant factor bi-criteria approximation for trees.} and consequently
how to use R\"acke's~(\citeyear{Racke}) tree decomposition technique.  Krauthgamer et al. (2014) give an $O(\log n)$ bi-criteria 
approximation algorithm for the problem: the algorithm finds 
a partition $P_1,\dots, P_k$ such that $|P_i| \leq O(\rho_i n)$ for every $i$ and the number of cut edges is $O(\log n \,OPT)$.
The algorithm first solves a configuration linear program and then uses a new
sophisticated method to round the fractional solution.

In this paper, we present a rather simple SDP based $O(\sqrt{\log n \log k})$ bi-criteria approximation algorithm for the problem. 
We note that our approximation guarantee matches that of the algorithm of
\citet*{KNS} for the the Minimum $k$-Partitioning problem (which is a special case of Minimum Nonuniform Partitioning, 
see above).
Our algorithm uses a technique of ``orthogonal separators'' developed by~\citet*{CMM} and later used by 
Bansal, Feige, Krauthgamer, Makarychev, Nagarajan, Naor, and Schwartz~(\citeyear{BFK}) for the Small Set Expansion problem. 
Using orthogonal separators, it is relatively easy to get a distribution over 
partitions $\{P_1,\dots, P_k\}$ such that $\Exp [|P_i|]\leq O(\rho_i n)$ for all $i$ and the expected number of
cut edges is $O(\sqrt{\log n \log (1/\rho_{min})}\,OPT)$ where $\rho_{min} = \min_i \rho_i$. The problem is that for some $i$, $P_i$ may be much larger than its expected size. 
The algorithm of Krauthgamer et al. (2014) solves a similar problem by first simplifying the instance and then grouping parts $P_i$ into ``mega-buckets''. We propose a simpler fix: Roughly speaking, 
if a set $P_i$ contains too many vertices, we remove some of these vertices and re-partition the removed vertices into $k$ pieces again. Thus we ensure that all capacity constraints
are (approximately) satisfied. It turns out that every vertex gets removed a constant number of times in expectation. Hence,
the re-partitioning step increases the number of cut edges 
only by a constant factor. 
Another problem is that $1/\rho_{min}$ may
be much larger than $k$. To deal with this problem, we transform the SDP solution
(eliminating ``short'' vectors)
and redefine thresholds $\rho_i$ so that
$1/\rho_{min}$ becomes $O(k)$.  

Our technique is quite robust and allows us to solve more general versions of the problem, \textit{Nonuniform Graph Partitioning with unrelated weights} and \textit{Nonuniform Graph Partitioning with unrelated $d$-dimensional weights}. 

Minimum Nonuniform Graph Partitioning with unrelated weights captures the variant of the problem where we assign vertices (tasks/jobs) to unrelated machines and the weight
of a vertex (the size of the task/job) depends on the machine it is assigned to.
\begin{definition}[Minimum Nonuniform Graph Partitioning with unrelated weights]
We are given a graph $G = (V, E)$ on $n$ vertices and a natural number $k \geq 2$. Additionally, we are given $k$ normalized
measures $\mu_1,\dots, \mu_k$ on $V$ (satisfying $\mu_i(V)=1$) and $k$ numbers $\rho_1, \dots, \rho_k \in (0,1)$. 
Our goal is to partition the graph into $k$ pieces (bins) $P_1$, \dots, $P_k$ such that $\mu_i(P_i) \leq \rho_i$ so as to minimize the number of cut edges. 
Some pieces $P_i$ may be empty.
\end{definition}

We will only consider instances of Minimum Nonuniform Graph Partitioning that have a feasible solution. 
We give an $O_{\varepsilon}(\sqrt{\log n \log \min(1/\rho_{min}, k)})$ bi-criteria approximation algorithm for the problem.
\begin{theorem}\label{thm:intro}
For every $\varepsilon >0$, there exists a randomized polynomial-time algorithm that given an instance of
Minimum Nonuniform Graph Partitioning with unrelated weights finds a partition $P_1,\dots, P_k$ satisfying $\mu_i(P_i)\leq 5(1+\varepsilon)\rho_i$.
The expected cost of the solution is at most $D\times OPT$, where $OPT$ is the optimal value, 
$D=O_{\varepsilon}(\sqrt{\log n \log \min(1/\rho_{min}, k)})$ and $\rho_{min}=\min_i \rho_i$.
For graphs with excluded minors $D=O_{\varepsilon}(1)$.
\end{theorem}

Nonuniform Graph Partitioning with unrelated $d$-dimensional weights
further generalizes the problem. In this variant of the problem,
we assume that we have $d$ resources 
(e.g. CPU speed, random access memory, disk space, network). Each piece $P_i$ has $c_j(i)$ units of resource $j\in\{1,\dots, d\}$, 
and each vertex $u$ requires $r_j(u,i)$ units of resource $j\in\{1,\dots, d\}$ when it is assigned to piece $P_i$.
We need to partition the graph so that capacity constraints for all resources are satisfied. 
The $d$-dimensional version of Minimum (uniform) $k$-Partitioning was previously studied by~\cite{AFK}. In their problem,
all $\rho_i=1/k$ are the same, and $r_j$'s do not depend on $i$.
\begin{definition}[Minimum Nonuniform Graph Partitioning with  unrelated $d$-dimensional weights]
We are given a graph $G = (V, E)$ on $n$ vertices. Additionally, we are given non-negative numbers $c_j(i)$ and $r_j(u,i)$
for $i\in \{1,\dots, k\}$, $j \in \{1,\dots,d\}$, $u \in V$.
Our goal is to find a partition of $V$ into $P_1,\dots, P_k$ subject to capacity constraints
$\sum_{u\in V} r_j(u,i) \leq c_j(i) \text{ for every } i \text{ and } j$
so as to minimize the number of cut edges.
\end{definition}
We present a bi-criteria approximation algorithm for this problem. 
\begin{theorem}\label{thm:intro-multi}
For every $\varepsilon >0$, there exists a randomized polynomial-time algorithm that given an instance of
Minimum Nonuniform Graph Partitioning with unrelated $d$-dimensional weights finds a partition $P_1,\dots, P_k$ satisfying 
$$\sum_{v\in V} r_j(v,i) \leq 5d(1 + \varepsilon)c_j(i) \quad\text{for every } i \text{ and } j.$$
The expected cost of the solution is at most $D\times OPT$, where $OPT$ is the optimal value, 
$D=O_{\varepsilon}(\sqrt{\log n \log k})$.
For graphs with excluded minors $D=O_{\varepsilon}(1)$.
\end{theorem}
We note that this result is a simple corollary
of Theorem~\ref{thm:intro} we let $\mu'_i (u) = \max_j (r_j(u,i)/c_j(i))$
and then apply our result to measures  $\mu_i(u) = \mu'_i(u)/\mu'_i(V)$ (we describe the details in Appendix~\ref{sec:multiple-resources}).

We remark that our algorithms work if edges in the graph have arbitrary positive weights. However, for simplicity of exposition,
we describe the algorithms for the setting where all edge weights are equal to one. To deal with arbitrary edge weights, we 
only need to change the SDP objective function.

Our paper strengthens the result of Krauthgamer et al. (\citeyear{KNST})
in two ways. First, it improves the approximation factor from $O(\log n)$ to $O(\sqrt{\log n \log k})$.
Second, it studies considerably more general 
variants of the problem, 
Minimum Nonuniform Partitioning with unrelated weights and
Minimum Nonuniform Partitioning with unrelated $d$-dimensional weights.
We believe that these variants are very natural. 
Indeed, one of the main motivations for the Minimum Nonuniform Partitioning problem is its applications to scheduling and load balancing: 
in these applications, the goal is to assign tasks to machines so as to minimize the total amount of communication between different machines, 
subject to capacity constraints. 
The constraints that we study in the paper are very general
and analogous to those that are often considered in the scheduling literature. 
We note that the method developed in  Krauthgamer et al. (2014) does not handle these more general variants of the problem. 
 
\section{Algorithm}
\subsec{SDP Relaxation}
Our relaxation for the problem is based on the SDP relaxation for the Small Set Expansion (SSE) problem
of \citet*{BFK}. We write the SSE relaxation for 
every cluster $P_i$ and then add consistency constraints similar to 
constraints used in Unique Games. 
For every vertex $u$ and index $i\in\{1,\dots, k\}$, we introduce a vector $\bar{u}_i$.
In the integral solution, this vector is simply the indicator variable for the event ``$u\in P_i$''.
It is easy to see that in the integral case, the number of cut edges equals~(\ref{sdp:objectiveA}).
Indeed, if $u$ and $v$ lie in the same $P_j$, then $\bar{u}_i = \bar{v}_i$ for all $i$; 
if $u$ lies in $P_{j'}$ and $v$ lies in $P_{j''}$ (for $j'\neq j''$) 
then $\|\bar u_i - \bar v_i\|^2 = 1$ for $i\in\{j',j''\}$ and 
$\|\bar u_i - \bar v_i\|^2 = 0$ for $i\notin\{j',j''\}$. The SDP objective is to minimize~(\ref{sdp:objectiveA}).

We add constraint (\ref{sdp:sum-u-in-P}) saying that $\mu_i(P_i)\leq \rho_i$. 
We further add spreading constraints (\ref{sdp:spread}) from \citet*{BFK} (see also \citet*{LM14}). 
The spreading constraints above are satisfied in the integral solution: If $u\notin P_i$,
then $\bar u_i= 0$ and both sides of the inequality equal 0. If $u\in P_i$,
then the left hand side equals $\mu_i(P_i)$, and the right hand side 
equals $\rho_i$. 

We write standard $\ell_2^2$-triangle inequalities (\ref{sdp:tr-ineq-1}) and (\ref{sdp:tr-ineq-2}).
Finally, we add consistency constraints. Every vertex $u$ must be assigned to one and only one $P_i$, hence
constraint (\ref{sdp:equal1}) is satisfied.
We obtain the following SDP relaxation.


\OpenFrame

\noindent \textbf{SDP Relaxation}
\begin{equation}\label{sdp:objectiveA}
\min \frac{1}{2}\sum_{i=1}^k \sum_{(u,v)\in E} \|\bar u_i - \bar v_i\|^2
\end{equation}
\noindent \textbf{subject to}
\begin{align}
\sum_{u\in V} \|\bar u_i\|^2 \mu_i(u) &\leq \rho_i \label{sdp:sum-u-in-P}&
\text{for all } i\in[k]\\
\sum_{v\in V} \langle \bar u_i, \bar v_i \rangle \mu_i(v)&\leq \|\bar u_i\|^2 \rho_i \label{sdp:size}\\
\text{for all } u\in V,\,i\in [k]\label{sdp:spread}\\
\sum_{i=1}^k \|\bar u_i\|^2 &= 1&\text{for all } u\in V\label{sdp:equal1}\\
\|\bar u_i-\bar v_i\|^2 + \|\bar v_i-\bar w_i\|^2
&\geq \|\bar u_i-\bar w_i\|^2&\text{for all } u,v,w\in V,\;i\in [k] \label{sdp:tr-ineq-1}\\
0\leq \langle \bar u_i, \bar v_i\rangle &\leq \|\bar u_i\|^2&\text{for all } u,v\in V,\;i\in [k]\label{sdp:tr-ineq-2} 
\end{align}


\CloseFrame

\subsec{Small Set Expansion and Orthogonal Separators}\label{sec:orthogonal}
Our algorithm uses a technique called ``orthogonal separators''. The notion of 
orthogonal separators was introduced in \citet*{CMM},
where it was used in an algorithm for Unique Games. Later,
\citet*{BFK} showed that the following holds. If the SDP solution satisfies constraints~(\ref{sdp:size}), (\ref{sdp:spread}), (\ref{sdp:tr-ineq-1}), and (\ref{sdp:tr-ineq-2}), then for every $\varepsilon\in (0,1)$,  $\delta\in (0,1)$, and $i\in [k]$, there exist
a distortion $D_i=O_{\varepsilon}(\sqrt{\log n \log (1/(\delta \rho_i))})$,
and a probability distribution over subsets of $V$ such that for a random set $S_i\subset V$ (``orthogonal separator'') distributed according to this distribution, we have  for 
$\alpha = 1/n$, 
\begin{itemize}
\item $\mu_i(S_i)\leq (1+\varepsilon)\rho_i$ (always);
\item For all $u$, 
$\Pr(u\in S_i)\in [(1-\delta) \alpha \|\bar u_i\|^2, \alpha \|\bar u_i\|^2]$;
\item For all $(u,v)\in E$, 
$\Pr(u\in S_i, v\notin S_i) \leq \alpha D_i \cdot \|\bar u_i- \bar v_i\|^2$.
\end{itemize}
We let $D = \max_{i} D_i$.
This statement was proved in \citet*{BFK} implicitly, so for completeness we prove it in the Appendix --- see Theorem~\ref{thm:orth-sep}.
For graphs with excluded minors and bounded genus graphs, $D =O_{\varepsilon}(1)$.

\subsec{Algorithm}
Let us examine a somewhat na{\"\i}ve algorithm for the problem inspired by the algorithm of \citet*{BFK} for  Small 
Set Expansion. We shall maintain the set of active (yet unassigned) vertices $A(t)$.
Initially, all vertices are active, i.e. $A(0)=V$. At every step $t$, we pick a random index $i \in \{1,\dots, k\}$ and sample 
an orthogonal separator $S_{i}(t)$ as described above. 
We assign all active vertices from 
$S_i(t)$ to the bin number $i$:
$$P_{i}(t+1)=P_{i}(t)\cup (S_{i}(t)\cap A(t)),$$
and mark all newly assigned vertices as inactive i.e., we let $A(t+1) = A(t)\setminus S_{i}(t)$.
We stop when the set of active vertices $A(t)$ is empty. We output the partition 
$\calP=\{P_1(T), \dots , P_k(T)\}$, where $T$ is the index of the  last iteration.

We can show that the number of edges cut by the algorithm is at most $O(D\times OPT)$, where $D$ is the distortion of orthogonal 
separators. Furthermore, the expected weight of each $P_i$ is $O(\rho_i)$. However, weights of some pieces may significantly deviate from
the expectation and may be much larger than $\rho_i$. So we need to alter the algorithm to guarantee that all sizes are bounded by $O(\rho_i)$
simultaneously. We face a problem similar to the one  \citet*{KNST} had to solve in their paper. Their solution is rather complex and does not seem to work in the weighted case. Here, we propose
a very simple fix for the na{\"\i}ve algorithm we presented above. We shall store vertices in every 
bin in layers. When we add new vertices to a bin at some iteration, we put them in a new layer on top of already stored vertices. Now, if
the weight of the bin number $i$ is greater than $5(1+\varepsilon) \rho_i$, we remove bottom layers from this bin so that its weight
is at most $5(1+\varepsilon) \rho_i$. Then we mark the removed vertices as active and jump to the next iteration. It is 
clear that this algorithm always returns a solution satisfying $\mu_i(P_i)\leq 5(1+\varepsilon) \rho_i$ for all $i$.
But now we need to prove that the algorithm terminates, and that the expected number of cut edges is still bounded 
by $O(D\times OPT)$. 

\tikzset{
    setbox/.style={
           rectangle,
           rounded corners=1pt,
           draw=darkgray, thin,
           minimum width=1.76cm,
           inner sep=0pt},
    limitline/.style={
           rectangle,
           minimum height=0pt,                      
           minimum width=1.76cm,
           inner sep=0pt,
           draw=blue!50!black,very thick}
           }

\begin{figure}[tb]
\scalebox{0.85}{
\begin{tikzpicture}
 \def\tmargin{0.3}
 \def\tspaceX{0.225}
 \def\tspaceY{0.1}

 \draw[black] (0,0)--(9,0);
 \draw[->,black] (0,0)--(0,4);
 \node[below left,text width=2cm,align=right] at (9,4) {\small
 partitioned vertices};
 
 \node[setbox,minimum height=12,above right] (a1) at (\tmargin,\tmargin) {};
 \node[setbox,minimum height=5,above=\tspaceY of a1] (a2) {};
 \node[setbox,minimum height=8,above=\tspaceY of a2] (a3) {};
 \node[setbox,minimum height=10,above=\tspaceY of a3] (a4) {};

 \node[below=\tmargin+\tspaceY of a1] (P1) {$P_1(t)$};
 \node[limitline,above=6*\tspaceY of a4,color=blue!50!black] (alimit) {};
 \node[above= 0 of alimit] {\footnotesize{$5(1+\varepsilon) \rho_1$}};
 
 \node[setbox,minimum height=3,right=\tspaceX of a1.south east, anchor=south west] (b1){};
 \node[setbox,minimum height=5,above=\tspaceY of b1] (b2) {};
 \node[setbox,minimum height=4,above=\tspaceY of b2] (b3) {};
 \node[setbox,minimum height=3,above=\tspaceY of b3] (b4) {};
 \node[setbox,minimum height=4,above=\tspaceY of b4] (b5) {};

 \node[below=\tmargin+\tspaceY of b1] (P2) {$P_2(t)$};
 \node[limitline,above=2.5*\tspaceY of b5] (blimit) {};
 \node[above= 0 of blimit] {\footnotesize{$5(1+\varepsilon) \rho_2$}};

 \node[setbox,minimum height=10,right=3*\tspaceX of b1.south east, anchor=south west,black!40!blue] (c1){};
 \node[setbox,minimum height=12,above=\tspaceY of c1,black!40!blue] (c2) {};
 \node[setbox,minimum height=6,above=\tspaceY of c2] (c3) {}; 
 \node[setbox,minimum height=8,above=\tspaceY of c3] (c4) {};
 \node[setbox,minimum height=15,thick,black,above=15*\tspaceY of c4,inner sep = 2pt] (cT) {\small $S_i(t)\cap A(t)$};

 \node[below=\tmargin+\tspaceY of c1] (Pi) {$P_i(t)$};
 \node[limitline,above=3*\tspaceY of c4] (climit) {};
 \node[above=0 of climit] (cTtext){\footnotesize{$5(1+\varepsilon) \rho_i$}};
 \node[below= 3pt of cT.south](cTsouthAnchor) {};
 \draw[->,black] (cTsouthAnchor.north)--(cTtext.north); 
 
   \node[setbox,minimum height=12,below=\tspaceY of Pi,opacity=0] (c2phantom) {\footnotesize{reactivated}}; 
   \node[setbox,minimum height=10,below=\tspaceY of c2phantom,opacity=0] (c1phantom) {\footnotesize{reactivated}};

 \node[setbox,minimum height=6,right=3*\tspaceX of c1.south east, anchor=south west] (d1){};
 \node[setbox,minimum height=7,above=\tspaceY of d1] (d2){};
 \node[setbox,minimum height=4,above=\tspaceY of d2] (d3) {};
 \node[setbox,minimum height=6,above=\tspaceY of d3] (d4) {}; 
 \node[setbox,minimum height=5,above=\tspaceY of d4] (d5) {};
 
 \node[below=\tmargin+\tspaceY of d1] (Pk) {$P_k(t)$};
 \node[limitline,above=1.2*\tspaceY of d5] (dlimit) {};
 \node[above= 0 of dlimit] {\footnotesize{$5(1+\varepsilon) \rho_k$}};

 \node at ($(P2)!0.5!(Pi)$) {$\dots$}; 
 \node at ($(Pi)!0.5!(Pk)$) {$\dots$}; 
  
\end{tikzpicture}
}
\hspace{0.45cm}
\scalebox{0.85}{
\begin{tikzpicture}
 \def\tmargin{0.3}
 \def\tspaceX{0.225}
 \def\tspaceY{0.1}

 \draw[black] (0,0)--(9,0);
 \draw[->,black] (0,0)--(0,4);

 \node[below left,text width=2cm,align=right] at (9,4) {\small
 partitioned vertices};
 
 \node[setbox,minimum height=12,above right] (a1) at (\tmargin,\tmargin) {};
 \node[setbox,minimum height=5,above=\tspaceY of a1] (a2) {};
 \node[setbox,minimum height=8,above=\tspaceY of a2] (a3) {};
 \node[setbox,minimum height=10,above=\tspaceY of a3] (a4) {};

 \node[below=\tmargin+\tspaceY of a1] (P1) {$P_1(t+1)$};
 \node[limitline,above=6*\tspaceY of a4] (alimit) {};
 \node[above= 0 of alimit] {\footnotesize{$5(1+\varepsilon) \rho_1$}};
 
 \node[setbox,minimum height=3,right=\tspaceX of a1.south east, anchor=south west] (b1){};
 \node[setbox,minimum height=5,above=\tspaceY of b1] (b2) {};
 \node[setbox,minimum height=4,above=\tspaceY of b2] (b3) {};
 \node[setbox,minimum height=3,above=\tspaceY of b3] (b4) {};
 \node[setbox,minimum height=4,above=\tspaceY of b4] (b5) {};

 \node[below=\tmargin+\tspaceY of b1] (P2) {$P_2(t+1)$};
 \node[limitline,above=2.5*\tspaceY of b5] (blimit) {};
 \node[above= 0 of blimit] {\footnotesize{$5(1+\varepsilon) \rho_2$}};
 
 \node[setbox,minimum height=6,right=3*\tspaceX of b1.south east, anchor = south west] (c3) {}; 
 \node[setbox,minimum height=8,above=\tspaceY of c3] (c4) {};
 \node[setbox,minimum height=15,thick,black,above=\tspaceY of c4,inner sep = 2pt] (cT) {\small $S_i(t)\cap A(t)$};
 
 \node[setbox,minimum height=10,above=\tspaceY of c4,opacity=0] (c1phantom) {};
 \node[setbox,minimum height=12,above=\tspaceY of c1phantom,opacity=0] (c2phantom) {};
  
 \node[below=\tmargin+\tspaceY of c3] (Pi) {$P_i(t+1)$};
 
  \node[setbox,minimum height=12,below=\tspaceY of Pi,black!40!blue] (c2){\color{black}\footnotesize{reactivated}}; 
  \node[setbox,minimum height=10,below=\tspaceY of c2,black!40!blue] (c1) {\color{black}\footnotesize{reactivated}};

 \node[limitline,above=3*\tspaceY of c2phantom] (climit) {};
 \node[above=0 of climit] (cTtext){\footnotesize{$5(1+\varepsilon) \rho_i$}};
 
 \node[setbox,minimum height=6,right=3*\tspaceX of c3.south east, anchor=south west] (d1){};
 \node[setbox,minimum height=7,above=\tspaceY of d1] (d2){};
 \node[setbox,minimum height=4,above=\tspaceY of d2] (d3) {};
 \node[setbox,minimum height=6,above=\tspaceY of d3] (d4) {}; 
 \node[setbox,minimum height=5,above=\tspaceY of d4] (d5) {};
 
 \node[below=\tmargin+\tspaceY of d1] (Pk) {$P_k(t+1)$};
 \node[limitline,above=1.2*\tspaceY of d5] (dlimit) {};
 \node[above= 0 of dlimit] {\footnotesize{$5(1+\varepsilon) \rho_k$}};

 \node at ($(P2)!0.5!(Pi)$) {$\dots$}; 
 \node at ($(Pi)!0.5!(Pk)$) {$\dots$}; 
  
\end{tikzpicture}
}
\caption{
The figure shows how we update sets $P_i(t)$ in iteration $t$. 
In this figure, rectangles represent layers of vertices in sets $P_i(t)$ (on the left) and $P_i(t+1)$ (on the right). All vertices 
in these layers are inactive (they are already  partitioned). Blue horizontal lines show capacity constraints.
In the example shown in the figure,  we add set $S_i(t) \cap A(t)$ to $P_i(t)$. The measure of the obtained set is greater than
$5(1+\varepsilon) \rho_i$, and so we remove the two bottom layers from $P_i(t) \cup (S_i(t) \cap A(t))$ 
(the removed layers are shown in blue). We get a set of measure at most $5(1+\varepsilon) \rho_i$. Vertices in the removed layers are reactivated after the iteration is over.
}
\label{fig:oneiteration}
\end{figure}
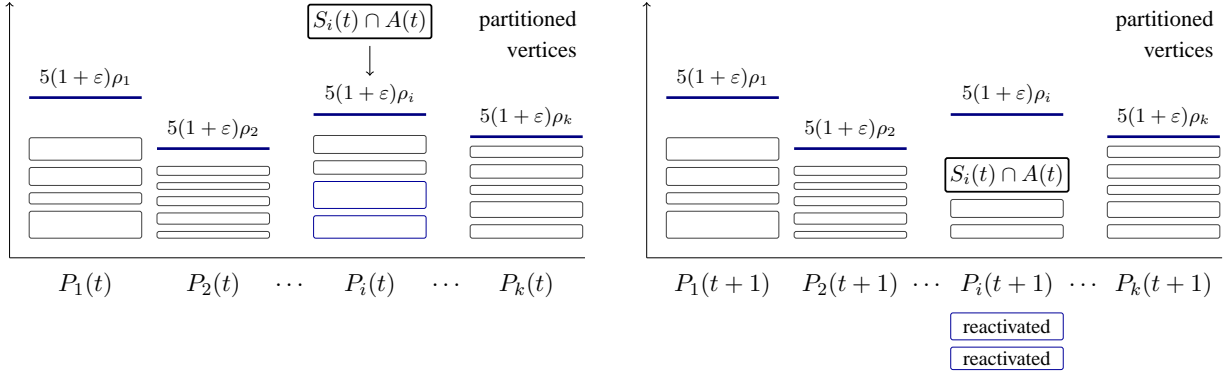

Before proceeding to the analysis, we describe the algorithm in detail.

\OpenFrame


\noindent \textbf{Algorithm for Nonuniform Partitioning with Unrelated Weights}

\vskip 1ex

\noindent \textbf{Input: }a graph $G=(V,E)$ on $n$ vertices; a positive integer $k\leq n$; a sequence of numbers $\rho_1,\dots, \rho_k\in (0,1)$ (with 
$\rho_1+\dots +\rho_k\geq 1$); weights $\mu_i:V\to \bbR^+$ (with $\mu_i(V)=1$).

\vskip 1ex

\noindent \textbf{Output: }a partitioning of vertices into disjoint sets $P_1,\dots, P_k$ such that $\mu_i(P_i)\leq 5(1+\varepsilon) \rho_i$.
\begin{itemize}
\item The algorithm maintains a partitioning of $V$ into a set of active vertices $A(t)$ and $k$ sets $P_1(t),\dots P_k(t)$, which we call bins.
For every inactive vertex $u\notin A(t)$, we remember its depth in the bin it belongs to. We denote the depth by $\depth_u (t)$. If $u\in A(t)$, then we let $\depth_u(t)=\perp$.
\item Initially, set $A(0) = V$; and $P_i(0)=\varnothing$, $\depth_u(t)=\perp$ for all $i$; $t = 0$.
\item \textbf{while} $A(t)\neq \varnothing$ 
\begin{enumerate}
\item Pick an index $i\in \{1,\dots, k\}$ uniformly at random.
\item Sample an orthogonal separator $S_{i}(t)\subset V$ with $\delta = \varepsilon/4$ as described in Section~\ref{sec:orthogonal}.
\item Store all active vertices from the set $S_i(t)$ in the bin number $i$. If $\mu_i(P_i(t) \cup (S_{i}(t)\cap A(t)))\leq 5(1+\varepsilon)\rho_i$, then 
simply add these vertices to $P_i(t+1)$:
$$P_i(t+1) = P_i(t) \cup (S_{i}(t)\cap A(t)).$$
Otherwise, find the largest depth $d$ such that $\mu_i(P_i(t+1))\leq 5(1+\varepsilon)\rho_i$, where
$$P_i(t+1) = \{u\in P_i(t): \depth_u(t)\leq d\} \cup (S_{i}(t)\cap A(t)).$$
In other words, add to the bin number $i$ vertices from $S_i(t)\cap A(t)$ and remove vertices from the bottom layers so that 
the weight of the bin is at most $5(1+\varepsilon)\rho_i$.
\item If we put at least one new vertex in the bin $i$ at the current iteration, that is, if $A(t)\cap S_i(t)\neq \varnothing$, then
set the depth of all newly stored vertices to 1; increase the depth of all other vertices in the bin $i$ by 1.
\item Update the set of active vertices: let $A(t+1)=V\setminus \bigcup_j P_j(t+1)$ and 
$\depth_u(t+1) = \perp$ for $u\in A(t+1)$. Let $t=t+1$.
\end{enumerate}
\item Set $T=t$ and return the partitioning $P_1(T),\dots, P_k(T)$.
\end{itemize}
\CloseFrame

\noindent Note that Step 3 is well defined. We can always find an index $d$ such that 
$\mu_i(P_i(t+1))\leq 5(1+\varepsilon)\rho_i$, because for $d=0$, we have
$P_i(t+1)= S_i(t) \cap A(t)$ and
thus
$$
\mu(P_i(t+1)) = 
\mu_i(S_i(t)\cap A(t))\leq \mu_i(S_i(t)) \leq (1+\varepsilon)\rho_i < 5(1+\varepsilon)\rho_i,$$
by the first property of orthogonal separators.

\subsec{Analysis}
We will first prove Theorem~\ref{thm:main} that states that the algorithm has approximation factor
$D=O_{\varepsilon}(\sqrt{\log n \log (1/\rho_{min})})$
on arbitrary graphs, and $D=O_{\varepsilon}(1)$ on graphs excluding a minor. Then we will show how to obtain 
$D = O_{\varepsilon}(\sqrt{\log n \log k})$ approximation on arbitrary graphs (see Appendix~\ref{sec:logk}). To this end, we will transform the SDP solution and redefine 
measures $\mu_i$ and capacities $\rho_i$ so that $\rho_{min} \geq \delta/k$, then apply Theorem~\ref{thm:main}.
The new SDP solution will satisfy all SDP constraints except possibly for constraint~(\ref{sdp:equal1}); 
it will however satisfy a relaxed constraint
\begin{equation}
\sum_{i=1}^k \|\bar u_i\|^2 \in [1-\delta,1]
\qquad\text{for all } u\in V.
\label{sdp:equal1approx}\tag{\ref{sdp:equal1}$'$}
\end{equation}
Thus in Theorem~\ref{thm:main}, we will assume only that the solution satisfies the SDP relaxation with constraint (\ref{sdp:equal1}) replaced by constraint (\ref{sdp:equal1approx}).

\begin{theorem}\label{thm:main}
The algorithm returns a partitioning $P_1(T),\dots, P_k(T)$ satisfying $\mu_i(P_i)\leq 5(1+\varepsilon)\rho_i$.
The expected number of iterations of the algorithm is at most $\Exp[T] \leq 4n^2 k + 1$ and the expected 
number of cut edges is at most $O(D\times SDP) = O(D\times OPT)$, where $D=O_{\varepsilon}(\sqrt{\log n \log (1/\rho_{min})})$
is the distortion of orthogonal separators; $\rho_{min}=\min_i \rho_i$. If the graph has an excluded minor, then $D=O_{\varepsilon}(1)$ (the constant depends on the excluded minor).

We assume only that the SDP solution given to the algorithm satisfies the SDP relaxation with constraint (\ref{sdp:equal1}) replaced by constraint (\ref{sdp:equal1approx}).
\end{theorem}
As we mentioned earlier, the algorithm always returns a valid partitioning. We need to verify that the algorithm terminates in 
expected polynomial time, and that it produces cuts of cost at most $O(D\times OPT)$ (see also Remark~\ref{rem:polytime}). 

The state of the algorithm at iteration $t$ is determined by the sets $A(t)$, $P_1(t),\dots, P_k(t)$ and the depths of the elements.
We denote the state by $\calC(t) = \{A(t), P_1(t),\dots, P_k(t), \depth(t)\}$. Observe that the probability that 
the algorithm is in the state $\calC^*$ at iteration $(t+1)$ is determined only by the state of the algorithm at 
iteration $t$. It does not depend on $t$ (given $\calC(t)$). So the
states of the algorithm form  
a Markov random chain.
The number of possible states is finite (since the depth of every vertex is bounded by $n$).
To simplify the notation, we assume that for $t\geq T$, $\calC(t)=\calC(T)$. This is consistent
with the definition of the algorithm --- if we did not stop the algorithm at time $T$, it would simply idle,
since $A(t)=\varnothing$, and thus $S_i(t)\cap A(t)=\varnothing$ for $t\geq T$.

We are interested in the probability that an inactive vertex $u$ which lies \textit{in the top layer} of one of 
the bins (i.e., $u\notin A(t)$ and 
$\depth_u(t)=1$) is removed from that bin within $m$ iterations. We let
$$
f(m, u, \calC^*) = 
\Pr(\exists t\in [t_0,t_0+m]\;\text{s.t.}\; u\in A(t)\Given \calC(t_0) = \calC^*, \depth_u(t_0)=1).
$$
That is, $f(m, u, \calC^*)$ is the probability that $u$ is removed from the bin $i$ at one of the iterations $t\in [t_0,t_0+m]$ given that 
at iteration $t_0$ the state of the algorithm is $\calC^*$ and $u$ is in the top layer of the bin $i$.
Note that the probability above does not depend on $t_0$ and thus $f(m, u, \calC^*)$ is well defined. We let
$$f(m) = \max_{u\in V}\max_{\calC^*} f(m , u, \calC^*).$$

Our fist lemma gives a bound on the expected number of steps on which a vertex $u$ is active in 
terms of $f(m)$.
\begin{lemma}\label{lem:bound-sum-A}
For every possible state of the algorithm $\calC^*$, every vertex $u$, and natural number $t_0$,
\begin{equation}\label{eq:sum-active}
\sum_{t=t_0}^{t_0+m} \Pr(u\in A(t)\Given \calC(t_0)=\calC^*)\leq \frac{k}{(1-2\delta)\alpha (1-f(m-1))}.
\end{equation}
\end{lemma}
\begin{proof}
The left hand side of inequality~(\ref{eq:sum-active}) equals expected number (conditioned on $\calC(t_0)=\calC^*$)
of iterations $t$ in the interval $[t_0,t_0+m]$ at which $u$ is active i.e., $u\in A(t)$. Our goal is to
upper bound this quantity.

Initially, at time $t_0$, $u$ is active or inactive. At every time $t$ when $u$ is active, $u$ 
is thrown in one of the bins $P_i$ with probability at least (here, we use that the SDP solution 
satisfies constraint (\ref{sdp:equal1approx}))
$$\frac{1}{k}\sum_{i=1}^k (1-\delta) \alpha \|\bar u_i\|^2 \geq \frac{(1-2\delta)\alpha}{k}.$$
So the expected number of iterations passed since $u$ becomes active till $u$ is stored in one of the bins and thus becomes inactive 
is at most $k/((1-2\delta)\alpha)$. 

Suppose that $u$ is stored in a bin $i$ at iteration $t$, then $u\in P_i(t+1)$ and $\depth_u(t+1)=1$. Thus, the probability that $u$ is reactivated  till 
iteration $t_0+m$ i.e., the  probability that for some $\tau \in [(t+1), t_0+m]\subset [(t+1), (t+1) + (m-1)]$, $u\in A(\tau)$ is at most $f(m-1)$.
Consequently, the expected number of iterations $t\in [t_0,t_0+m]$ at which $u$ is active is bounded by 
$$\frac{k\cdot 1}{(1-2\delta)\alpha} + \frac{k\cdot f(m)}{(1-2\delta)\alpha}  + 
\frac{k\cdot f^2(m)}{(1-2\delta)\alpha}  + \cdots = \frac{k}{(1-2\delta)\alpha (1-f(m))}.$$
\end{proof}

We now show that $f(m)\leq 1/2$ for all $m$.
\begin{lemma}\label{lem:returnProb} For all natural $m$, $f(m) \leq 1/2$.
\end{lemma}
\begin{proof}
We prove this lemma by induction on $m$. For $m=0$, the statement is trivial as $f(0)=0$. 

Consider an arbitrary state $\calC^*$, bin $i^*$, vertex $u$, and iteration $t_0$. Suppose that $\calC(t_0)=\calC^*$, $u\in P_{i^*}(t_0)$ and $\depth_u(t_0)=1$ i.e., $u$ lies in the top
layer in the bin $i^*$. We need to estimate the probability that $u$ is removed from the bin $i^*$ till iteration $t_0+m$. The vertex $u$ is removed 
from the bin $i^*$ if and only if at some iteration $t\in \{t_0, \dots ,t_0+m-1\}$, $u$ is ``pushed away'' from the bin  by new vertices 
(see Step 2 of the algorithm). This happens only if the weight of vertices added to the bin $i^*$ at iterations 
$\{t_0, \dots ,t_0+m-1\}$ plus the weight of vertices in the first layer of the bin at iteration $t_0$ exceeds $5(1+\varepsilon)\rho_i$. 
Since the weight of vertices in the first layer is at most $(1+\varepsilon) \rho_i$, the weight of vertices added to the bin $i^*$ at iterations $\{t_0, \dots ,t_0+m-1\}$ must be greater than $4(1+\varepsilon)\rho_{i^*}$.

We compute the expected weight of vertices thrown in the bin $i^*$ at iterations $t\in \{t_0, \dots ,t_0+m-1\}$. 
Let us introduce some notation: $M=\{t_0, \dots ,t_0+m-1\}$; $i(t)$ is the index $i$ chosen by the algorithm at the iteration $t$. Let $X_{M, i^*}$ be the 
weight of vertices thrown in the bin $i^*$ at iterations $t\in M$. Then,
\begin{align}
\Exp \big[X_{M, i^*}\given \calC(t_0)=\calC^*\big] &= \label{eq:sum-bin-i-vert} \Exp\Big[\sum_{\substack{t\in M\\s.t.\;i(t)=i^*}} \mu_{i^*}\big(S_{i^*}(t)\cap A(t)\big)
\given \calC(t_0)=\calC^*\Big]\\
&= \sum_{t\in M}\sum_{v\in V}\Pr\big(i(t)=i^* \text{ and } v\in S_{i^*}(t)\cap A(t)\Given \calC(t_0)=\calC^*\big) \mu_{i^*}(v).
\notag
\end{align}
The event ``$i(t)=i^*$ and $v\in S_{i^*}(t)$'' is independent from the event ``$v\in A(t)$ and $\calC(t_0)=\calC^*$''.
Thus,
\begin{multline*}
\Pr\big(i(t)=i^*\text{ and } v\in S_{i^*}(t)\cap A(t)\Given \calC(t_0)=\calC^*\big)\\
 = \Pr\big(i(t)=i^*\text{ and } v\in S_{i^*}(t)\big)\cdot \Pr\big(v\in A(t)\Given \calC(t_0)=\calC^*\big).
\end{multline*}
Since $i(t)$ is chosen uniformly at random in $\{1,\dots,k\}$, we have $\Pr(i(t)=i^*)=1/k$. Then, by 
property 2 of orthogonal separators, $\Pr(v\in S_{i^*}(t)\Given i(t)=i^*) \leq \alpha \|\bar v_{i^*}\|^2$.
We get
\[
\Pr\big(i(t)=i^*\text{ and } v\in S_{i^*}(t)\cap A(t)\Given \calC(t_0)=\calC^*\big) 
\leq \frac{\alpha \|\bar v_{i^*}\|^2}{k}\cdot \Pr\big(v\in A(t)\Given \calC(t_0)=\calC^*\big).
\]
We now plug this expression in~(\ref{eq:sum-bin-i-vert}) and use Lemma~\ref{lem:bound-sum-A},
\begin{align*}
\Exp [X_{M, i^*}\given \calC(t_0)=\calC^*]&\leq 
\sum_{v\in V} \frac{\alpha \|\bar v_{i^*}\|^2\mu_{i^*}(v)}{k}\cdot \sum_{t\in M} \Pr\big(v\in A(t)\Given \calC(t_0)=\calC^*\big)\\
&\leq 
\sum_{v\in V} \frac{\alpha \|\bar v_{i^*}\|^2\mu_{i^*}(v)}{k}\cdot \frac{k}{(1-2\delta)\alpha (1-f(m-1))}\\
&= \sum_{v\in V} \frac{\|\bar v_{i^*}\|^2\mu_{i^*}(v)}{(1-2\delta)(1-f(m-1))}.
\end{align*}
Finally, observe that $1-f(m-1)\geq 1/2$ by the inductive hypothesis, and 
$\sum_{v\in V} \|\bar v_{i^*}\|^2\mu_{i^*}(v) \leq \rho_{i^*}$ by the SDP 
constraint~(\ref{sdp:sum-u-in-P}).
Hence, $\Exp [X_{M, i^*}\given \calC(t_0)=\calC^*] \leq 2\rho_{i^*}/(1-2\delta)$. By Markov's inequality,
$$\Pr\big(X_{M, i^*} \geq 4(1+\varepsilon) \rho_{i^*} \big)\leq 
\frac{2\rho_{i^*}}{4(1-2\delta)(1+\varepsilon) \rho_{i^*}}\leq \frac{1}{2},$$
since $\delta = \varepsilon/4$. This concludes the proof.
\end{proof}

As an immediate corollary of Lemmas~\ref{lem:bound-sum-A} and~\ref{lem:returnProb}, we get that for all $u\in V$,
\begin{equation}\label{eq:expected-active-iter}
\sum_{t=0}^{\infty} \Pr(u\in A(t))= 
\lim_{m\to \infty}\sum_{t=0}^{m} \Pr(u\in A(t))\leq \frac{2k}{(1-2\delta)\alpha}\leq \frac{4k}{\alpha}.
\end{equation}
\begin{proof}[Proof of Theorem~\ref{thm:main}]
We now prove Theorem~\ref{thm:main}. We first bound the expected running time.
At every iteration of the algorithm $t<T$, the set $A(t)$ is not empty. Hence,
using~(\ref{eq:expected-active-iter}), we get
$$\Exp [T]\leq \Exp\Big[\sum_{t=0}^{\infty}|A(t)|\Big] + 1= \sum_{v\in V}\sum_{t=0}^{\infty}\Pr(v\in A(t)) + 1
\leq n\cdot \frac{4k}{\alpha} + 1 = 4kn^2 + 1.$$

We now upper bound the expected size of the cut. For every edge $(u,v)\in E$ we estimate the probability that $(u,v)$ is cut. Suppose that $(u,v)$ is cut. Then, $u$ and $v$ belong to distinct sets $P_i(T)$. Consider
the iteration $t$ at which $u$ and $v$ are separated the first time. A priori, 
there are two possible cases:
\begin{enumerate}
\item At iteration $t$, $u$ and $v$ are active, but only one of the vertices $u$ or $v$ is added to some 
set $P_i(t+1)$; the other vertex remains in the set $A(t+1)$.
\item At iteration $t$, $u$ and $v$ are in some set $P_i(t)$, but only one of the vertices $u$ or $v$ is
removed from the set $P_i(t+1)$.
\end{enumerate}
It is easy to see that, in fact, the second case is not possible, since if $u$ and $v$ were never
separated before iteration $t$, then $u$ and $v$ must have the same depth (i.e., $\depth_u(t) = \depth_v(t)$)
and thus $u$ and $v$ may be removed from the bin $i$ only together. 

Consider the first case, and assume that $u\in P_{i(t)}(t+1)$ and $v\in A(t+1)$. Here,
as in the proof of Lemma~\ref{lem:returnProb}, we denote the index $i$ chosen at iteration $t$ by $i(t)$.
Since $u\in P_{i(t)}(t+1)$ and $v\in A(t+1)$, we have $u\in S_{i(t)}(t)$ and $v\notin S_{i(t)}(t)$.
Write
\begin{align*}
\Pr(u,v\in A(t);&\;u\in S_{i(t)}(t);\;v\notin S_{i(t)}(t))=\\&=
\Pr(u,v\in A(t))\cdot\Pr(u\in S_{i(t)}(t);\;v\notin S_{i(t)}(t))\\
&= \Pr(u,v\in A(t)) \cdot \sum_{i=1}^k \frac{\Pr(u\in S_i(t);\;v\notin S_i(t)\given i(t)=i)}{k}.
\end{align*}
We replace $\Pr(u,v\in A(t))$ with $\Pr(u\in A(t))\geq \Pr(u,v\in A(t))$, and then 
use the inequality $\Pr (u\in S_i(t);\;v\notin S_i(t))\leq \alpha D\;\|\bar u_i -\bar v_i\|^2$,
which follows from the third property of orthogonal separators. We get
\[
\Pr(u,v\in A(t);\;u\in S_{i(t)}(t);\;v\notin S_{i(t)}(t)) 
\leq
\Pr(u\in A(t))\times \Big(\frac{1}{k}\sum_{i=1}^k \alpha D\;\|\bar u_i -\bar v_i\|^2\Big).
\]
Thus, the probability that $u$ and $v$ are separated at iteration $t$ is upper bounded by 
$\Big(\Pr(u\in A(t)) + \Pr(v\in A(t))\Big)\times \Big(\frac{1}{k}\sum_{i=1}^k \alpha D\;\|\bar u_i -\bar v_i\|^2\Big).$
The probability that the edge $(u,v)$ is cut (at some iteration) is at most
\begin{multline*}
\Big(\sum_{t=0}^{\infty} \Pr(u\in A(t)) + \Pr(v\in A(t))\Big)\times \Big(\frac{1}{k}\sum_{i=1}^k \alpha D\;\|\bar u_i -\bar v_i\|^2\Big)\leq \\
\leq \frac{8k}{\alpha} \Big(\frac{1}{k}\sum_{i=1}^k \alpha D\;\|\bar u_i -\bar v_i\|^2\Big)
=8 \sum_{i=1}^k D\;\|\bar u_i -\bar v_i\|^2.
\end{multline*}
To bound the first term on the left hand side we used inequality~(\ref{eq:expected-active-iter}). 
We get the desired bound on the expected number of cut edges:
$$\sum_{(u,v)\in E}\Pr ((u,v)\text{ is cut})\leq 8 \sum_{(u,v)\in E} \sum_{i=1}^k D\;\|\bar u_i -\bar v_i\|^2
= 16 D \cdot SDP,$$
where $SDP$ is the SDP value.
\end{proof}

\appendix
\section{Orthogonal Separators}
For completeness, we prove Theorem~\ref{thm:orth-sep}.

\begin{theorem}[\citet*{BFK}]\label{thm:orth-sep}
There exists a polynomial-time algorithm that given a graph $G=(V,E)$, a measure $\mu$ on $V$ ($\mu(V)=1$), 
parameters $\rho, \varepsilon, \delta \in (0,1)$ and a collection of vectors $\bar{u}$
satisfying the following constraints:
\begin{align}
\sum_{u\in V} \|\bar u\|^2 \mu(u) &\leq \rho&
\text{for all } i\in[k]\\
\sum_{v\in V} \langle \bar u, \bar v \rangle \mu(v)&\leq \|\bar u\|^2 \rho&
\text{for all } u\in V,\,i\in[k]\label{A-sdp:spread-constr}\\
\|\bar u-\bar v\|^2 + \|\bar v-\bar w\|^2
&\geq \|\bar u-\bar w\|^2&\text{for all } u,v,w\in V,\;i\in[k]\\
0\leq \langle \bar u, \bar v\rangle &\leq \|\bar u\|^2&\text{for all } u,v\in V,\;i\in[k]\label{A-sdp:triang-ineq}\\
\|\bar u \|^2 &\leq 1& \text{for all } u\in V,\,i\in[k]
\end{align}
outputs a random set $S\subset V$ (``orthogonal separator'') such that
\begin{enumerate}
\item $\mu(S)\leq (1+\varepsilon)\rho$ (always);
\item For all $u$, 
$\Pr(u\in S)\in [(1-\delta) \alpha \|\bar u\|^2, \alpha \|\bar u\|^2]$;
\item For all $(u,v)\in E$, 
$\Pr(u\in S, v\notin S) \leq \alpha D \cdot \|\bar u-\bar v\|^2$.
\end{enumerate}
Where the probability scale $\alpha = 1/n$, and the distortion $D\leq O_{\varepsilon}(\sqrt{\log n \log (1/(\rho\delta))})$. For graphs with 
excluded minors, $D=O_{\varepsilon}(1)$.
\end{theorem}

In \citet*{CMM}, we showed that there exists a randomized polynomial-time algorithm that outputs a random set $S$ 
with the following properties (see also~\citet*{BFK} and ~\citet*{LM14}):
\begin{itemize}
\item For all $u\in V$, $\Pr (u\in S) = \alpha \, \| \bar u\|^2$.
\item For all $u,v\in V$ with  $\|\bar u - \bar v\|^2\geq \beta
\min (\| \bar u\|^2, \|\bar v\|^2)$,
$$\Pr (u \in S \text{ and } v \in S) \leq \frac{\alpha \min(\|\bar u \|^2, \|\bar v \|^2)}{m}.$$
\item For all $(u,v)\in E$,
$$\Pr(u\in S\text{ and } v\notin S)\leq \alpha  D \times \|\bar u - \bar v\|^2.$$
\end{itemize}
Here $m > 0$ is a parameter of the algorithm; $\alpha = 1/n$ is a probability scale; 
$D\leq O_{\beta}(\sqrt{\log n \log m})$ is the distortion.
\citet*{BFK} showed that for graphs with excluded minors, $D=O(1)$.

Our algorithm samples $S$ as above (with $m=2/(\delta\varepsilon\rho)$, $\beta = \varepsilon/4$) and outputs $S'=S$ if $\mu(S)\leq (1 +\varepsilon)\rho$,
and $S'=\varnothing$, otherwise.
It is clear that $\mu(S')\leq (1 + \varepsilon)\rho$ (always), and thus the first property in Theorem~\ref{thm:orth-sep} is satisfied. Then, for $(u,v)\in E$,
$$\Pr(u\in S'\text{ and } v\notin S')\leq \Pr(u\in S\text{ and } v\notin S) \leq \alpha  D \times \|\bar u - \bar v\|^2,$$
where $D=O_{\beta}(\sqrt{\log n \log m}) = O_{\varepsilon}(\sqrt{\log n \log (1/(\rho\delta))})$.

For every $u\in V$,
$$\Pr(u\in S')\leq \Pr(u\in S)= \alpha \|\bar u \|^2.$$
So we only need to verify that $\Pr(u\in S')\geq \alpha (1-\delta) \|\bar u \|^2$. We assume $\|\bar{u}\|^2 \neq 0$. We have 
$$\Pr(u\in S') = \Pr(u\in S'\Given u\in S)\cdot \Pr (u\in S) = \Pr\big(\mu(S)\leq (1 + \varepsilon)\rho \Given u\in S\big)\cdot \alpha \|\bar u \|^2.$$
We split $V$ into two sets $A_u = \{v: \|\bar u - \bar v\|^2 \geq \beta \|\bar u\|^2\}$ and 
$B_u = \{v: \|\bar u - \bar v\|^2 < \beta \|\bar u\|^2\}$. We show below (see Lemma~\ref{lem:sizeA}) that $\mu(B_u)\leq (1 + \varepsilon/2)\rho$. 
Then, 
$$\mu(S)=\mu(S\cap A_u)+\mu(S\cap B_u)\leq (1+\varepsilon /2)\rho + \mu (S\cap A_u)$$
and 
\begin{equation}\label{eq:bound-u-in-S}
\Pr(u\in S') \geq
\alpha \|\bar u \|^2\cdot\Pr \big(\mu(S\cap A_u)\leq \varepsilon\rho/2\Given u\in S\big).
\end{equation}
We estimate $\Pr \big(\mu(S\cap A_u)\geq \varepsilon \rho /2\given u\in S)$. For every $v\in A_u$, $\|\bar u - \bar v\|^2 \geq \beta \|\bar u\|^2$.
Thus, for $v\in A_u$, $\Pr (u \in S; v\in S)\leq \alpha \|\bar u \|^2/m$, and 
$$\Pr (v\in S\Given u \in S) = \frac{\Pr(u \in S, v \in S)}{\Pr(u\in S)}\leq \frac{1}{m}.$$
Therefore, $\Exp[\mu(S\cap A_u)\given u\in S]\leq \mu(A_u)/m\leq 1/m$, and, by Markov's inequality,
$$\Pr\big(\mu(S\cap A_u)\geq \varepsilon\rho/2 \Given u\in S\big)\leq 
\frac{\Exp[\mu(S)\given u\in S]}{\varepsilon\rho/2}\leq \frac{2}{m\varepsilon\rho}\leq \delta.$$
We plug this bound in~(\ref{eq:bound-u-in-S}) and get the desired bound,
$$\Pr(u\in S')\geq \alpha \|\bar u \|^2\cdot\Pr \big(\mu(S\cap A_u)\leq \varepsilon \rho/2\Given u\in S\big)
\geq \alpha \|\bar u \|^2\cdot (1-\delta).$$

We now prove Lemma~\ref{lem:sizeA}.
\begin{lemma}\label{lem:sizeA} For every $u\in S$ with $\|\bar u\|^2\neq 0$,
$\mu(B_u)\leq (1+\varepsilon/2)\rho$.
\end{lemma}
\begin{proof}
If $v\in B_u$, then by the definition of $B_u$ and by inequality~(\ref{A-sdp:triang-ineq}), we have
$$\|\bar u\|^2 - \langle \bar u, \bar v \rangle =
\|\bar u - \bar v\|^2 - (\|\bar v\|^2 - \langle \bar u, \bar v \rangle) \leq 
\beta \|\bar u\|^2.$$
Thus, $\langle \bar u, \bar v \rangle\geq (1-\beta) \|\bar u\|^2$. Now, we use constraint~(\ref{A-sdp:spread-constr}),
\begin{align*}
\mu(B_u)&= \sum_{v\in B_u}\mu (v) \leq \sum_{v\in B_u} \mu(v) \cdot \frac{\langle \bar u, \bar v \rangle}{(1-\beta)\|\bar u\|^2}
\leq \frac{1}{(1-\beta)\|\bar u\|^2} \sum_{v\in V}  \langle \bar u, \bar v \rangle \mu(v)\\
&\leq \frac{(1+2\beta)}{\|\bar u\|^2}\cdot \rho \|\bar u \|^2 = (1+2\beta)\rho = (1+\varepsilon/2)\rho.
\end{align*}
\end{proof}

\section{$O(\sqrt{\log n \log k})$ approximation}
\label{sec:logk}
\begin{theorem}
There is a polynomial-time randomized algorithm that returns a partitioning $P_1(T),\dots, P_k(T)$ satisfying $\mu_i(P_i)\leq 5(1+\varepsilon)\rho_i$
such that the expected 
number of cut edges is at most $O(D\times OPT)$, where $D=O_{\varepsilon}(\sqrt{\log n \log k})$.
\end{theorem}
\begin{proof}
We perform three steps. First we solve the SDP relaxation, then transform its solution and change measures $\mu_i$, and finally apply Theorem~\ref{thm:main}
to the obtained SDP solution. 

We start with describing how we transform the solution.
We set $\delta = \varepsilon/4$ as before. Then we choose a threshold $\theta$ uniformly at random from 
$[\delta/2,\delta]$. We let $\tilde u_i = \bar u_i$ if $\|\bar u_i\|^2 \geq \theta / k$ and 
$\tilde u_i = 0$, otherwise. It is immediate that the solution $\tilde u_i$ satisfies all SDP constraints except possibly constraint~(\ref{sdp:equal1}). Note, however, that it satisfies constraint (\ref{sdp:equal1approx}):
$$\sum_{i=1}^k \|\tilde u_i\|^2 = 
\sum_{i=1}^k \|\bar u_i\|^2 - \sum_{i: \|\tilde u_i\|^2 < \theta/k} \|\tilde u_i\|^2 
= 1 -   \sum_{i: \|\tilde u_i\|^2 < \theta/k} \|\tilde u_i\|^2 \in [ 1 - \delta, 1].
$$
Consider two vertices $u$ and $v$. Assume without loss of generality that $\|\bar u_i\|^2 \leq \|\bar v_i\|^2$.
If either  
$\|\bar u_i\|^2 \leq \|\bar v_i\|^2 < \theta / k$ or 
$\theta / k\leq \|\bar u_i\|^2 \leq \|\bar v_i\|^2$,
then we have $\|\tilde u_i - \tilde v_i\| = \|\bar u_i - \bar v_i\|$. Otherwise, 
if $\|\bar u_i\|^2 < \theta/k \leq \|\bar v_i\|^2$, we have 
$$\|\tilde u_i - \tilde v_i\| = \|\bar v_i\|^2 \leq  \|\bar u_i - \bar v_i\|^2  + \|\bar u_i\|^2 = 
\|\bar u_i - \bar v_i\|^2  + \delta/k.$$
Therefore,
$$
\Exp[\|\tilde u_i - \tilde v_i\|^2] \leq \|\bar u_i - \bar v_i\|^2 + (\delta/k)\Prob{\|\bar u_i\|^2 < \theta/k \leq \|\bar v_i\|^2}. $$
To upper bound $\Prob{\|\bar u_i\|^2 < \theta/k \leq \|\bar v_i\|^2}$, note that the random variable $\theta$ is distributed uniformly on $(\delta/2, \delta)$, so its probability density is bounded from above by $2/\delta$. We get from SDP constraint (\ref{sdp:tr-ineq-2}) that
$\|\bar v_i\|^2 - \|\bar u_i\|^2 \leq \|\bar u_i - \bar v_i\|^2$.
Thus, 
$$\Prob{\|\bar u_i\|^2 < \theta/k \leq \|\bar v_i\|^2}
\leq  (2k/\delta)\cdot \|\bar u_i - \bar v_i\|^2,$$
We have,
$$
\Exp[\|\tilde u_i - \tilde v_i\|^2]
\leq  \|\bar u_i - \bar v_i\|^2 + (\delta/k) \cdot (2k/\delta)\cdot  \|\bar u_i- \bar v_i\|^2 = 3 \|\bar u_i - \bar v_i\|^2.
$$
We conclude that the SDP value of solution $\tilde u_i$ is at most $3 SDP \leq 3 OPT$ in expectation.

Now we modify measures $\mu_i$ and capacities $c_i$.
Let $A_i = \{u:\bar u_i \neq 0\}$. 
Define 
\begin{align*}
\mu'_i(Z) &= \mu_i(Z\cap A_i)/\mu_i(A_i) \text{ for } Z\subseteq V,\\ 
\tilde \rho_i &= \rho_i / \mu_i(A_i)
\end{align*}
(if $\mu_i(A_i) = 0$ we let $\tilde\mu_i = \mu_i$ and $\tilde\rho_i =1$, essentially removing the capacity constraint for $P_i$).
We have $\tilde\mu_i(V) = \mu_i(A_i)/\mu_i(A_i) = 1$.
By (\ref{sdp:sum-u-in-P}), we get
$$
\rho_i \geq \sum_{u\in V} \|\bar u_i\|^2 \mu_i(u) \geq 
\sum_{u\in A_i} \|\bar u_i\|^2 \mu_i(u) \geq 
\sum_{u\in A_i} \frac{\delta}{2k} \cdot \mu_i(u) = \frac{\delta\mu_i(A)}{2k}.$$
Therefore, $\tilde\rho_i = \rho_i / \mu_i (A) \geq \delta/(2k)$,
and $\tilde\rho_{min}= \min \tilde\rho_i \geq \delta/(2k)$
(if $\mu_i(A_i) = 0$ then $\tilde\rho_i = 1 > \delta/(2k)$).

Note that since each $\rho_i$ increases by a factor of $1/\mu_i(A_i)$ 
and each $\mu_i(u)$ increases by a factor at most $1/\mu_i(A_i)$,
vectors $\tilde u_i$ satisfy SDP constraints (\ref{sdp:sum-u-in-P}) and (\ref{sdp:spread}), in which
$\mu_i$ and $\rho_i$ are replaced with $\tilde\mu_i$ and $\tilde\rho_i$, respectively
(assuming that $\mu_i(A_i) \neq 0$; if $\mu_i(A_i) = 0$, the constraints clearly hold).
We run the algorithm from Theorem~\ref{thm:main} on vectors
$\tilde u_i$ with measures $\tilde\mu_i$ and capacities $\rho_i$.
The algorithm finds a partition $P_1,\dots, P_k$ that cuts at most $D \cdot SDP \leq D\cdot OPT$ edges, where
$D = O_{\varepsilon}(\sqrt{\log n\log (1/\tilde\rho_{min})}) =  
O_{\varepsilon}(\sqrt{\log n\log k})$. We verify that the weight
of each set $P_i$ is $O(\rho_i)$. Note that $P_i \subset A_i$
since for $u\notin A_i$, $\|\tilde u_i\|^2 = 0$, and thus the
algorithm does not add $u$ to $P_i$. We have,
$$\mu_i(P_i) = \mu_i'(P_i\cap A_i)\cdot \mu_i(A_i) = \mu_i'(P_i)\cdot \mu_i(A_i) \leq 5(1+\varepsilon)\tilde\rho_i \cdot \mu_i(A_i) \leq  5(1+\varepsilon)\rho_i.$$
\end{proof}

\section{Partitioning with $d$-Dimensional Weights}
\label{sec:multiple-resources}
We describe how Minimum Nonuniform Graph Partitioning with unrelated $d$-dimensional weights reduces to  Minimum Nonuniform Graph Partitioning with unrelated weights.
Consider an instance $\calI$ of Minimum Nonuniform Graph Partitioning with unrelated $d$-dimensional weights.
Let $\mu'_i (u) = \max_j (r_j(u,i)/c_j(i))$. Then define measures  $\mu_i(u)$ and capacities $\rho_i(u)$ by
$$\mu_i(u) = \mu'_i(u)/\mu'_i(V) \quad\text{and}\quad  \rho_i = d/\mu'_i(V).$$
We obtain an instance $\calI'$.
Note that the optimal solution $P^*_1, \dots, P^*_k$ for  $\calI$ is a feasible solution for $\calI'$ since
\begin{align*}
\mu_i(P^*_i) &= \sum_{u\in P^*_i} \frac{\mu'_i(u)}{\mu'_i(V)} = \frac1{\mu'_i(V)}\sum_{u\in P^*_i} \max_j \frac{r_j(u,i)}{c_j(i)}
   \leq  \frac1{\mu'_i(V)} \sum_{u\in P^*_i} \sum_{j=1}^d \frac{r_j(u,i)}{c_j(i)}\\
     &=   \frac1{\mu'_i(V)} \sum_{j=1}^d \sum_{u\in P^*_i}\frac{r_j(u,i)}{c_j(i)}
     \leq \frac{d}{\mu'_i(V)} = \rho_i.
\end{align*}

We solve instance $\calI'$ and get a partitioning $P_1, \dots, P_k$
that cuts at most $O(\sqrt{\log n \log k}\, OPT)$ edges.
The partitioning satisfies $d$-dimensional capacity constraints:
\[
\sum_{u\in P_i} r_j(u,i) \leq \sum_{u\in P_i} c_j(i) \mu'_i(u)
= c_j(i) \mu'_i(V)  \sum_{u\in P_i} \mu_i(u) \leq \\ \leq 
c_j(i) \mu'_i(V) (5 (1+\varepsilon) \rho_i) = 5 d (1+\varepsilon)
\, c_j(i).
\]

This concludes the analysis of the reduction.

\begin{remark}\label{rem:polytime}
The algorithm $\calA$ from Theorem~\ref{thm:main} is a randomized algorithm: it always finds a feasible solution 
(a solution with $|P_i| \leq 5(1+\varepsilon) \rho_i$), 
the expected cost of the solution is  $\alpha_\calA SDP = O(D \times OPT)$ (where $\alpha_\calA = O(D)$), and the expected number of iterations the algorithm performs
is upper bounded by $4n^2k +1$. The algorithm can be easily converted to an algorithm $\calA'$ that always runs in polynomial-time and
that succeeds with high probability. If it succeeds, it outputs a feasible solution of cost $O(D\times OPT)$; if it fails,
it outputs $\perp$ ($\perp$ is a special symbol that indicates that the algorithm failed).
The algorithm $\calA'$ works as follows. It executes $\cal A$. If $\calA$ does not stop after $(4n^4 k + n^2)$ iterations, $\calA'$ terminates and outputs $\perp$.
Otherwise, it compares the value of the solution that $\calA$ found with $3\alpha_\calA \, SDP$: If the cost is less than $3\alpha_\calA \, SDP$, the algorithm 
outputs the solution; otherwise it outputs $\perp$. Clearly the algorithm always runs in polynomial time, and if it succeeds
it finds a solution of cost at most $3\alpha_\calA\, OPT = O(D\times OPT)$. By Markov's inequality,
the probability that the algorithm fails is at most $1/n^2 + 1/3 < 1/2$. By running the algorithm $n$ times, we can make the failure
probability exponentially small (note that we need the algorithm to succeed at least once).
\end{remark}

\end{document}